\documentclass[12pt,leqno]{article}
\usepackage{amsthm,amsmath}
\usepackage{latexsym}
\usepackage{amssymb}
\usepackage[utf8]{inputenc}
\usepackage[a4paper,left=2.0cm,right=2.0cm,top=2.0cm,bottom=2.0cm]{geometry}
\theoremstyle{plain}
\newtheorem{Thm}{Theorem}

\newtheorem{Prop}[Thm]{Proposition}
\newtheorem{Lem}[Thm]{Lemma}
\theoremstyle{definition}

\renewcommand{\phi}{\varphi}

\newcommand{\RR}{\mathbb{R}}

\begin{document}
\title{A new product on $2\times 2$ matrices}
\author{L. Kramer,  P. Kramer, and V. Man'ko}
\maketitle

\begin{abstract}
We study a bilinear multiplication rule on $2\times 2$ matrices which is intermediate between the ordinary
matrix product and the Hadamard matrix product, and we relate this to the hyperbolic motion group of the plane.
\end{abstract}

% \section{}

The observables of many systems in physics are treated by sets of matrices.
The elements of these matrices are taken from the field $F$ of real or complex numbers.
A multiplicative and bilinear composition rule  
$ (a,b)\longmapsto a*b$ of matrices provides  the system with the structure of an algebra $\bf A$.
The standard composition rule for two matrices is matrix 
multiplication. This composition is given  by  (row by column) multiplication. 
By contrast,  the Hadamard product of matrices \cite{1} is defined by the entry-wise multiplication of elements. 
We shall study and interpret a composition rule intermediate between the standard and the Hadamard case
which appears in \cite{0}. We refer to \cite{3} for further discussions of matrix products.

\section{Properties of the matrix algebra and the $\star$-product.}

Suppose that $F$ is a field. We will be mainly interested in the case where $F$ is the field of real numbers $\bf R$,
but most our results are valid in general (for example, $F$ could also denote the field of complex numbers).
The multiplicative group of the field will be denoted by $F^\times$.
We recall that an algebra is a vector space $\bf A$ over $F$ with a bilinear product $*$ defined on pairs of vectors.
Hence we require for all vectors $u,v,w\in \bf A$ and all scalars $s\in F$ that 
\begin{gather*}
 (u+v)*w=u*w+v*w\\
 w*(u+v)=w*u+w*v\\
 (su)*v=u*(sv)=s(u*v).
\end{gather*}
The algebra $\bf A$ is called associative if 
\[
 u*(v*w)=(u*v)*w
\]
holds for all vectors $u,v,w$.
For example, the $n\times n$-matrices with the usual matrix product form an associative algebra $\mathrm{Mat}(n,F)$, whose 
unit element is the identity matrix $\bf1$.
But other products on matrices have also been studied. 

In this note we 
propose a new product $\star$ on $2\times 2$ matrices and study some of its properties.
The product $\star$ is defined as
\[\tag{1}
 \left(
 \begin{array}{ll}
  a_{11} & a_{12}\\
  a_{21} & a_{22}
\end{array} 
\right)\star 
\left(
 \begin{array}{ll}
  b_{11} & b_{12}\\
  b_{21} & b_{22}
\end{array} 
\right)
=
\left(
 \begin{array}{cc}
  a_{11}b_{11} & a_{11}b_{12}+a_{12}b_{22}\\
  a_{21}b_{11}+a_{22}b_{21} & a_{22}b_{22}
\end{array} 
\right).
\]
In order to study this product in a systematic way we set up the following notation.
We denote $n\times n$-matrices by capital letters $A,B,C$. The identity matrix is denoted by
$\bf 1$ and the zero matrix is denoted by $\bf0$.
Every $n\times n$ square matrix $A$ can be decomposed uniquely as a sum \[A=A_0+A_1\] of a diagonal matrix
$A_0$ and a matrix $A_1$ with zeros on the diagonal.
We define a new product on $n\times n$-matrices by putting 
\[\tag{2}
 A\star B=A_0B_0+(AB)_1.
\]
For the products on the right-hand side we use ordinary matrix multiplication.
Hence entries on the diagonal of $A\star B$ are computed by multiplying diagonal entries, while
the off-diagonal entries are computed in the ordinary matrix multiplication  way.
The star product thus mixes the ordinary matrix product with the Hadamard product.
It is clear from the definition that the product $\star$ is bilinear,
\[A\star (B+C)=A\star B+A\star C\text{ and }(B+C)\star A=B\star A+C\star A.\]
It is also clear that the product $\star$ is for $n=1$ the ordinary multiplication of scalars. 
For $n=2$, this product given in (2) coincides with the one given in equation (1).
\begin{Lem}
The identity matrix $\bf1$ is a unit element, ${\bf1}\star A=A=A\star{\bf1}$.
The product $\star$ is associative for $n=2$, but not associative for $n\geq 3$.
For $n=2$ we have 
\[\tag{3}
 A\star B=A_0B_0+A_0B_1+A_1B_0.
\]
\end{Lem}
\begin{proof}
The fact that $\bf1$ is a unit is clear from formula (2). 
We always have \[(AB)_1=A_1B_0+A_0B_1+(A_1B_1)_1,\]
because $(A_0B_0)_1=\bf0$.
If $n=2$, then the product of two
matrices with zeros on the diagonal is a diagonal matrix, whence $(A_1B_1)_1={\bf0}$ in this case. 
Therefore we have for $n=2$ the formula
\[
 A\star B=A_0B_0+A_0B_1+A_1B_0.
\]
and thus 
\begin{multline*}
 (A\star B)\star C=(A_0B_0+A_0B_1+A_1B_0)\star C
 \\
 =A_0B_0C_0+A_0B_0C_1+A_0B_1C_0+A_1B_0C_0=A\star(B\star C).
\end{multline*}
For $n=3$ we have, however,
\[
 \left(\left(\begin{array}{lll}0&1&0\\0&0&0\\0&0&0\end{array}\right)
 \star 
 \left(\begin{array}{lll}0&0&0\\0&0&1\\0&0&0\end{array}\right)\right)
 \star
 \left(\begin{array}{lll}0&0&0\\0&0&0\\0&1&0\end{array}\right)=
 \left(\begin{array}{lll}0&1&0\\0&0&0\\0&0&0\end{array}\right)
\] 
and
\[
 \left(\begin{array}{lll}0&1&0\\0&0&0\\0&0&0\end{array}\right)
 \star 
 \left(\left(\begin{array}{lll}0&0&0\\0&0&1\\0&0&0\end{array}\right)
 \star
 \left(\begin{array}{lll}0&0&0\\0&0&0\\0&1&0\end{array}\right)\right)=
 \left(\begin{array}{lll}0&0&0\\0&0&0\\0&0&0\end{array}\right) \rlap,
\] 
which shows that $\star$ is not associative. The same example works for all $n\geq 3$ by
extending these three $3\times 3$ matrices with zeros to $n\times n$ matrices.
\end{proof}

\section{The group of $\star$-invertible matrices}
Now we study the $\star$-invertible matrices, for $n=2$.
\begin{Lem}
 Suppose that $n=2$. Then $A=A_0+A_1$ is $\star$-invertible if and only of $A_0$ is invertible
 in the ordinary sense. The $\star$-inverse $B$ of $A=A_0+A_1$ is then $B=A_0^{-1}-A_0^{-1}A_1A_0^{-1}$. 
\end{Lem}
\begin{proof}
If $A_0$ is invertible, we have 
 \[
  (A_0+A_1)\star(A_0^{-1}-A_0^{-1}A_1A_0^{-1})={\bf1}+A_1A_0^{-1}-A_1A_0^{-1}={\bf1}.
 \]
and similarly $(A_0^{-1}-A_0^{-1}A_1A_0^{-1})\star(A_0+A_1)={\bf1}$.
Hence $A=A_0+A_1$ is $\star$-invertible with $\star$-inverse $B=A_0^{-1}-A_0^{-1}A_1A_0^{-1}$.

If $A_0$ is not invertible, let $B_0$ denote the diagonal matrix where the two diagonal entries of $A_0$ are
exchanged. Then $A_0B_0=\bf0$ and the matrix $B=B_0-A_1$ satisfies 
\[
 A\star B=A_0B_0-A_0A_1+A_1B_0={\bf0},
\]
hence $A$ cannot be $\star$-invertible if $A_0$ is not invertible.
\end{proof}
We let $G$ denote the group of all $\star$-invertible matrices of our algebra, for $n=2$.
Every $\star$-invertible element is of the form
\[\tag{3}
 A=A_0+A_1=A_0({\bf 1}+A_0^{-1}A_1)=A_0\star({\bf1}+B_1),\text{ where }B_1=A_0^{-1}A_1.
\]
Let $D$ denote the set of all $2\times 2$ diagonal matrices.
On this set $D$ of diagonal matrices, the $\star$-product and the usual matrix product coincide.
The $\star$-invertible matrices in $D$ thus form a commutative subgroup $H$ of $G$.
Let $N$ denote the set of all matrices of the form
$B={\bf1}+B_1$. These matrices are $\star$-invertible and they form 
a group, with group law \[({\bf1}+B_1)\star({\bf1}+C_1)={\bf1}+B_1+C_1.\]
Hence the group $N$ is also commutative and isomorphic to the additive group $E=F\times F$.
\begin{Lem}
The group $G$ is the semidirect product of $H$ and the invariant subgroup $N$,
\[
 G=HN=H\ltimes N.
\]
\end{Lem}
\begin{proof}
The subgroups $H,N\subseteq G$ have obviously trivial intersection $H\cap N=\{{\bf1}\}$
and we noted above in equation (3) that $G=HN$.
For $A=A_0$ in $H$ and $B={\bf1}+B_1$ in $N$ we have 
\[A\star B=A_0+A_0B_1=({\bf 1}+A_0B_1A_0^{-1})\star A_0=\tilde B\star A,\] which shows that $N\unlhd G$
is an invariant subgroup in $G$. Hence $G=H\ltimes N$.
\end{proof}

\section{A geometric interpretation of the group $\boldsymbol G$}
For this last section we assume that 
\[
 1+1\neq 0,
\]
which is certainly true for the case that $F$ is the field of real numbers $\bf R$.
We denote by $\mathrm{SO}(1,1)$ the group of all $2\times 2$ matrices of determinant $1$ which leave the bilinear form 
\[
 b(u,v)=u_1v_1-u_2v_2
\]
on the $2$-dimensional vector space $E=F\times F$ invariant. 
This group is abelian and consists of all matrices of the form
\[\tag{4}
R=\left(
 \begin{array}{ll}
  c & s\\
  s & c
\end{array} 
\right)
\quad\text{ with }\quad c^2-s^2=1.
\]
For the case $F=\bf R$ the group $\mathrm{SO}(1,1)$ has a subgroup $\mathrm{SO}^+(1,1)$ of index $2$ consisting of all matrices of the form 
\[\tag{5}
\phi(t)=\left(
 \begin{array}{ll}
  \cosh(t) & \sinh(t)\\
  \sinh(t) & \cosh(t)
\end{array} 
\right)
\quad\text{ with }t\in\RR,
\]
and the one parameter group $t\longmapsto \phi(t)$ is a Lie group isomorphism 
\[
\phi:{\bf R}\xrightarrow{\ \cong\ } \mathrm{SO}^+(1,1).
\]
In general. the map $\psi:F^\times\longrightarrow\mathrm{SO}(1,1)$ that maps the nonzero scalar $x$ to the matrix
\[
 \psi(x)=\left(
 \begin{array}{ll}
  \frac{1}{2}(x+x^{-1}) & \frac{1}{2}(x-x^{-1})\\
  \frac{1}{2}(x-x^{-1}) & \frac{1}{2}(x+x^{-1})
\end{array} 
\right)
\]
is a group isomorphism $\psi:F^\times\xrightarrow{\ \cong\ }\mathrm{SO}(1,1)$.

The hyperbolic motion group $\mathrm{ISO}(1,1)$ consists of all affine transformations
$ T_{[R,u]}$ of the $2$-dimensional plane $E=F\times F$ of the form
\[
T_{[R,u]}(x)=Rx+u,
\]
where $R$ is a $2\times2$ matrix in $\mathrm{SO}(1,1)$ and $u$ is a vector in $E$.
Such a transformation consists thus of a hyperbolic rotation $R$ followed by a translation by a vector $u$.
The group law on $\mathrm{ISO}(1,1)$ is thus
\[\tag{6}
T_{[R,u]}*T_{[S,v]}=T_{[RS,u+Rv]}.
\]

We define three auxiliary maps $\alpha,\beta,\gamma$ as follows.
We put 
\[
 \alpha
 \left(
 \begin{array}{ll}
  x & 0 \\
  0 & y
\end{array} 
\right)=
\left(
 \begin{array}{ll}
  \frac{1}{2}(\frac{x}{y}+\frac{y}{x}) & \frac{1}{2}(\frac{x}{y}-\frac{y}{x})\\
  \frac{1}{2}(\frac{x}{y}-\frac{y}{x}) & \frac{1}{2}(\frac{x}{y}+\frac{y}{x})
\end{array} 
\right).
\]
Then $\alpha:H\longrightarrow\mathrm{SO}(1,1)$ is a surjective group homomorphism whose kernel
consists of all diagonal matrices of the form $s\bf1$, with $s\neq 0$.
We also put 
\[
 \beta
 \left(
 \begin{array}{ll}
  1 & p \\
  q & 1
\end{array} 
\right)=
\left(
 \begin{array}{c}
  p+q\\p-q
\end{array} 
\right)
\]
and we note that $\beta$ is a group isomorphism $N\xrightarrow{\ \cong\ } E$.
Finally, we put 
\[
 \gamma
 \left(
 \begin{array}{ll}
  x & 0 \\
  0 & y
\end{array} 
\right)=x
\]
and we note that $\gamma:H\longrightarrow F^\times$ is a surjective group homomorphism whose
kernel consists of all diagonal matrices of the form
$\left(
 \begin{smallmatrix}
  1 & 0 \\
  0 & y
\end{smallmatrix} 
\right)$.
We define a map $\Phi:G\longrightarrow \mathrm{ISO}(1,1)$ by
\[\tag{7}
\Phi(A)=\Phi(A_0+A_1)=T_{[\alpha(A_0),\beta({\bf1}+A_1A_0^{-1})]}
=T_{[\alpha(A_0),\beta(AA_0^{-1})]}.
\]
\begin{Prop}
The map \[\Phi:G\longrightarrow \mathrm{ISO}(1,1)\]
is a surjective group homomorphism whose kernel consists of all matrices of the form 
$s{\bf1}$ with $s\neq 0$.
\end{Prop}
\begin{proof}
We have to verify that $\Phi(A\star B)=\Phi(A)*\Phi(B)$ holds for all $A,B\in G$.
If $A=A_0$ and $B=B_0$, then 
\[
\Phi(A)*\Phi(B)=T_{[\alpha(A_0),0]}*T_{[\alpha(B_0),0]}=T_{[\alpha(A_0)\alpha(B_0),0]}=T_{[\alpha(A_0B_0),0]}=\Phi(A\star B). 
\]
Similarly, if $A={\bf1}+A_1$ and $B={\bf1}+B_1$, then 
\[
\Phi(A)*\Phi(B)=T_{[{\bf1},\beta(A)]}*T_{[{\bf1},\beta(B)]}=T_{[{\bf1},\beta(A)+\beta(B)]}=T_{[{\bf1},\beta(AB)]}=\Phi(A\star B).
\]
Hence the map $\Phi$ is a group homomorphism both on $H$ and on $N$.

Since $G=HN$ is a semidirect product, it remains to show that 
$\Phi(A\star B)=\Phi(A)*\Phi(B)$ holds for all matrices $A,B$ of the form 
$A=A_0\in H$ and $B={\bf1}+B_1\in N$.
We put
\[A=\left(\begin{matrix}
  x & 0 \\
  0 & y
\end{matrix} \right)
\quad\text{ and }\quad
B=\left(
 \begin{matrix}
  1 & p \\
  q & 1
\end{matrix} 
\right).
\]
Then $\alpha(A)=
\left(
 \begin{smallmatrix}
  c & s \\
  s & c
\end{smallmatrix} 
\right)$,
where $c=\frac{1}{2}(\frac{x}{y}+\frac{y}{x})$ and $s=\frac{1}{2}(\frac{x}{y}-\frac{y}{x})$.
We also put $z=c+s=\frac{x}{y}$, and we compute
\[
 \Phi(A)*\Phi(B)=T_{\left[\left(
 \begin{smallmatrix}
  c & s\vphantom{0} \\
  s & c
\end{smallmatrix} \right),\left(\begin{smallmatrix} 0 \\ 0 \end{smallmatrix}\right)\right]}
*
T_{\left[\left(
 \begin{smallmatrix}
  1 & 0 \\
  0 & 1
\end{smallmatrix} \right),\left(\begin{smallmatrix} p+q\\p-q\end{smallmatrix}\right)\right]}=
 T_{\left[\left(
 \begin{smallmatrix}
  c & s\vphantom{pz^{-1}} \\
  s & c\vphantom{pz^{-1}}
\end{smallmatrix} \right),\left(\begin{smallmatrix} zp+z^{-1}q\\ zp-z^{-1}q\end{smallmatrix}\right)\right]}.
\]
On the other hand, $A\star B=\left(\begin{smallmatrix} x & 0 \\ 0 & y \end{smallmatrix}\right)+
\left(\begin{smallmatrix} 0 & xp \\ yq & 0 \end{smallmatrix}\right)$.
Thus 
\[
 \Phi(A\star B)=T_{\left[\left(
 \begin{smallmatrix}
  c & s \vphantom{pz^{-1}} \\
  s & c \vphantom{pz^{-1}}
\end{smallmatrix} \right), \left(\begin{smallmatrix} zp+z^{-1}q\\ zp-z^{-1}q\end{smallmatrix}\right)\right]}.
\]
The kernel of $\Phi$ consists of the multiples $s\bf 1$ of the identity matrix.
This finishes the proof of the proposition.
\end{proof}

By a well-known construction, see \cite{2}, one can describe the $2$-dimensional affine transformations by $3\times 3$ matrices. We introduce
a third coordinate which is set to $1$.
Then $T_{[R,u]}$ corresponds to the $3\times 3$~matrix 
\[
\widetilde T_{[R,u]}=\begin{pmatrix} r_{11} & r_{12} & u_1 \\ r_{21} & r_{22} & u_2 \\ 0 & 0 & 1\end{pmatrix}.
\]
From this we obtain a faithful $3$-dimensional representation $\rho$ of $G$ via 
\[\tag{8}
\rho\begin{pmatrix} x & p \\ q & y \end{pmatrix}  =
x\begin{pmatrix}
 c & s & zp+z^{-1}q\\
 s & c & zp-z^{-1}q \\
 0 & 0 & 1
\end{pmatrix},
\]
where again $c=\frac{1}{2}(\frac{x}{y}+\frac{y}{x})$ and $s=\frac{1}{2}(\frac{x}{y}-\frac{y}{x})$ and   $z=c+s=\frac{x}{y}$.

\section{Conclusion.}

We have shown that the matrix group $G$, equipped with the $\star$
multiplication,  is a semidirect product of two abelian groups. 

By a $3$-dimensional representation, the elements of the group $G$  can be  interpreted as
planar hyperbolic rotations, followed by  translations.  
The $\star$ multiplication of matrices can  be 
interpreted as the composition rule of these elements.

\end{document}